\documentclass[twocolumn,twosides]{IEEEtran} 

\makeatletter
\let\NAT@parse\undefined
\makeatother

\usepackage{amsmath,amssymb,amsfonts}
\usepackage[final]{graphicx}
\usepackage{flushend}
\usepackage[numbers,sort&compress]{natbib}

\usepackage{ucs} 
\usepackage[utf8x]{inputenc}
\usepackage[usenames,dvipsnames]{xcolor}
\usepackage{tikz}
\usepackage{tkz-tab}
\usepackage{multirow}
\usepackage{latexsym}
\usepackage{amssymb}
\usepackage{amsmath}

\makeatletter
\def\blfootnote{\xdef\@thefnmark{}\@footnotetext}
\makeatother

\newtheorem{corollary}{Corollary}
\newtheorem{theorem}{Theorem}

\newtheorem{lemma}{Lemma}
\newtheorem{definition}{Definition}

\def \M{\mathcal{M}}

\def \L {{\mathcal L }}

\def \K {{\mathcal K }}

\def\M{ {\mathcal{M}} }

\begin{document}
\title{\Huge The $\kappa$-$\mu$ Shadowed Fading Model\\ with Integer Fading Parameters}

\author{
\vspace{3mm}
\authorblockN{F. Javier Lopez-Martinez, Jose F. Paris and Juan M. Romero-Jerez}}


\maketitle
\begin{abstract}
\blfootnote{\noindent  This work has been submitted to the IEEE for possible publication. Copyright may be transferred without notice, after which this version may no longer be accesible.\\
\noindent F. J. Lopez-Martinez and J. F. Paris are with Departmento de Ingenier\'ia de Comunicaciones, Universidad de Malaga - Campus de Excelencia Internacional Andaluc\'ia Tech., Malaga 29071, Spain. Contact e-mail: fjlopezm@ic.uma.es. \\ \indent J. M. Romero-Jerez is with Departmento de Tecnolog\'ia Eectr\'onica, Universidad de Malaga - Campus de Excelencia Internacional Andaluc\'ia Tech., Malaga 29071, Spain.
}
We show that the popular and general $\kappa$-$\mu$ shadowed fading model with integer fading parameters $\mu$ and $m$ can be represented as a mixture of squared Nakagami-$\hat m$ (or Gamma) distributions. Thus, its PDF and CDF can be expressed in closed-form in terms of a finite number of elementary functions (powers and exponentials). The main implications arising from such connection are then discussed, which can be summarized as: (1) the performance evaluation of communication systems operating in $\kappa$-$\mu$ shadowed fading becomes as simple as if a Nakagami-$\hat m$ fading channel was assumed; (2) the $\kappa$-$\mu$ shadowed distribution can be used to approximate the $\kappa$-$\mu$ distribution using a closed-form representation in terms of elementary functions, by choosing a sufficiently large value of $m$; and (3) restricting the parameters $\mu$ and $m$ to take integer values has limited impact in practice when fitting the $\kappa$-$\mu$ shadowed fading model to field measurements. As an application example, the average channel capacity of communication systems operating under $\kappa$-$\mu$ shadowed fading is obtained in closed-form.
\end{abstract}

\vspace{0mm}

\IEEEpeerreviewmaketitle

\section{Introduction}

The $\kappa$-$\mu$ shadowed fading model was introduced in \cite{Paris2014}, and right after in \cite{Cotton2015} in an independent work, as a generalization of the popular $\kappa$-$\mu$ model proposed by Yacoub \cite{Yacoub07}. { Other generalizations of the $\kappa$-$\mu$ distribution using the inverse Gamma distribution \cite{Yoo2015}, the inverse Gaussian distribution \cite{Sofotasios2013}, and a multiplicative composition with the Gamma distribution \cite{Yoo2016}, are also available in the literature.} Recently, it was shown that the apparently unrelated $\eta$-$\mu$ distribution was also a particular case of the $\kappa$-$\mu$ shadowed distribution \cite{Laureano2016}, thus encapsulating this set of popular and general fading distributions in the literature in a single model. Ever since its inception, the $\kappa$-$\mu$ shadowed fading model has gained a remarkable attention in the literature due to its versatility on modeling propagation conditions ranging from very favorable to worse-than-Rayleigh fading. It also provides an good fit to field measurements in diverse scenarios like device-to-device communications \cite{Cotton2015}, or underwater acoustic channels \cite{Paris2014,Francis2016}.

Unlike other general fading models \cite{Durgin2002,Salo2007,Rao2015}, its chief probability functions (PDF and CDF) are given in closed-form. That being said, the computation of its CDF still requires for the evaluation of the confluent bivariate hypergeometric function $\Phi_2$ \cite[9.262.2]{Gradstein2007}, which is not yet included in commercial mathematical software packages. This fact has not prevented the widespread use of the $\kappa$-$\mu$ shadowed distribution in a number of practical scenarios of interest \cite{Kumar2015,Kalyani2015,Bhatnagar2015}. However, these results have a much more complicated form than their counterparts when assuming, for instance, the simpler and extremely popular Nakagami-$\hat m$ (or simply Nakagami) fading model\footnote{In order to avoid any confusion between the $m$ parameter of the $\kappa$-$\mu$ shadowed fading model and the homonymous parameter of the Nakagami-$\hat m$ fading model, we use a $m$ with a superscript to denote the latter.}. 

In this paper, we show that the $\kappa$-$\mu$ shadowed PDF and CDF can be expressed in terms of a finite number of elementary functions for a proper choice of the fading severity parameter values. Specifically, we show that the $\kappa$-$\mu$ shadowed distribution can be expressed as a mixture of squared Nakagami-$\hat m$ distributions, when the parameters $\mu$ and $m$ take integer values. As we will later see, such restriction has little effect in practice when fitting field measurements to the $\kappa$-$\mu$ shadowed distribution, while being extremely convenient from a computational perspective.

This connection considerably facilitates the performance evaluation of communication systems operating in $\kappa$-$\mu$ shadowed fading channels. In fact, we show that \textit{any performance metric} that is calculated by averaging over the distribution of the SNR in $\kappa$-$\mu$ shadowed fading (e.g. bit error rate, capacity, outage probability...) can be readily and directly obtained as a linear combination of the results obtained when assuming Nakagami-$\hat m$ fading; the values of the weights for this linear combination are the coefficients of the mixture, which are given in closed-form. We also show that the $\kappa$-$\mu$ shadowed fading model can be used to approximate the $\kappa$-$\mu$ distribution with arbitrary precision, by simply choosing a sufficiently large value of $m$. Thus, the computational benefits of the new representation of the $\kappa$-$\mu$ shadowed fading model in terms of elementary functions can be extended to the $\kappa$-$\mu$ distribution (and also to the Rician distribution as a special case for $\mu=1$).

As a direct application, we exemplify the usefulness of the results here unveiled to obtain exact expressions for the average capacity of the $\kappa$-$\mu$ shadowed fading channel, which are considerably simpler than those originally obtained in \cite{Celia2014}.

The remainder of this paper is structured as follows: in Section \ref{new}, the new expressions for the PDF and the CDF of the $\kappa$-$\mu$ shadowed fading model are given in terms of a finite sum of elementary functions. Then, a number of relevant implications and useful properties arising from these results are discussed in Section \ref{implications}. An application example is presented in Section \ref{apps}, whereas the main conclusions are outlined in Section \ref{conc}.

\section{New statistics for the $\kappa$-$\mu$ shadowed distribution.}
\label{new}

Throughout this paper, we will characterize the distribution of the received power envelope in $\kappa$-$\mu$ shadowed fading channels, or equivalently, the instantaneous SNR $\gamma$ at the receiver. Note that characterizing the distribution of the amplitude envelope $r$ is straightforward by a simple change of variables $r=\sqrt{\gamma}.$

\begin{definition}
A random variable $\gamma$ following a $\kappa$-$\mu$ shadowed distribution will be denoted as $\gamma  \sim \mathcal{S}\left( {\bar \gamma ;\kappa ,\mu ,m} \right)$, and its PDF will be given by
\begin{equation}
\begin{split}
\label{eq:001}
f_\mathcal{S}\left( {\bar \gamma ;\kappa ,\mu ,m;x} \right) & =
 \frac{\mu^\mu m^m (1+\kappa)^\mu }
  {\Gamma(\mu)(m+\mu\kappa)^m\overline{\gamma}} \left(\frac{x}{\overline{\gamma}}\right)^{\mu-1}
\\ \cdot & e^{-\frac{\mu(1+\kappa)}{\overline{\gamma}}x}
  {_1F_1}   \left(m, \mu ; \frac{\mu^2\kappa(1+\kappa)}{\mu\kappa+m} \frac{x}{\overline{\gamma}}\right),
\end{split}
\end{equation}
where ${_1F_1}(\cdot)$ is the confluent hypergeometric function of the first kind \cite[eq. (9.210.1)]{Gradstein2007}. 
\end{definition}
\vspace{3mm}
\begin{definition}
A random variable $\gamma$ following a $\kappa$-$\mu$ distribution will be denoted as ${\gamma  \sim \mathcal{KM}\left( {\bar \gamma ;\kappa ,\mu} \right)}$, and its PDF will be given by
\begin{equation}
\label{eq:001b}
\begin{split}
f_{\mathcal{KM}}(\gamma)=&\frac{\mu(1+\kappa)^\frac{\mu+1}{2}}{\bar{\gamma}\kappa^\frac{\mu-1}{2}e^{\mu\kappa}}\left(\frac{\gamma}{\bar{\gamma}}\right)^{\frac{\mu-1}{2}}\\
&\times e^{-\frac{\mu(1+\kappa)\gamma}{\bar{\gamma}}}I_{\mu-1}\left(2\mu\sqrt{\frac{\kappa(1+\kappa)\gamma}{\bar{\gamma}}}\right),
\end{split}
\end{equation}
where $I_\nu(\cdot)$ is the $\nu$-th order modified Bessel function of first kind.
\end{definition}
\vspace{3mm}
\begin{definition}
A random variable $\gamma$ following a squared Nakagami distribution will be denoted as $\gamma  \sim \K\left( {\bar \gamma ;\hat m} \right)$, and its PDF, assuming $\hat m \in {\mathbb N}$, will be given by
\begin{equation}
\begin{split}
\label{eq:002}
f_\K \left( {\bar \gamma ;\hat m;x} \right) = \left( {\frac{\hat m}
{{\bar \gamma }}} \right)^{\hat m} \frac{{x^{\hat{m} - 1} }}
{{\left( {\hat{m} - 1} \right)!}}e^{ - x\hat{m}/\bar \gamma}.
\end{split}
\end{equation}
Note that the squared Nakagami distribution is equivalent to a Gamma distribution with shape parameter $\hat m$ and scale parameter $\bar \gamma / \hat{m}$.
\end{definition}
\vspace{3mm}
After these preliminary definitions, we now provide new expressions for the PDF and CDF of the $\kappa$-$\mu$ shadowed distribution for positive integer values of the fading severity parameters $\mu$ and $m$.
\vspace{3mm}
\begin{theorem} \label{T1}
Let $\gamma$ be a random variable such that $\gamma  \sim \mathcal{S}\left( {\bar \gamma ;\kappa ,\mu ,m} \right)$ and let  $\mu,m \in {\mathbb N}$. Then, $\gamma$ is a mixture of squared Nakagami distributions, which PDF is given as follows:
\vspace{2mm}

$ \bullet$ If $m<\mu$
\begin{equation}
\label{eq:003}
\begin{split}
  & f_\mathcal{S} \left( {\bar \gamma ;\kappa ,\mu ,m;x} \right) = \sum\limits_{j = 1}^{\mu  - m} {A_{1j} } f_\K \left( {\omega _{A1} ;\mu  - m - j + 1;x} \right)  \cr 
  & \quad \quad \quad \quad  + \sum\limits_{j = 1}^m {A_{2j} } f_\K \left( {\omega _{A2} ;m - j + 1;x} \right). \cr
\end{split}
\end{equation}

$ \bullet$ If $m \geq \mu$
\begin{equation}
\label{eq:004}
\begin{split}
f_S \left( {\bar \gamma ;\kappa ,\mu ,m;x} \right) = \sum\limits_{j = 0}^{m - \mu } {B_j } f_K \left( {\omega _B ;m - j;x} \right),
\end{split}
\end{equation}
where
\begin{equation}
\label{eq:005}
\begin{split}
  & A_{1j}  =  \left( { - 1} \right)^m \left( {\begin{array}{c}
   {m + j - 2}  \\ 
   {j - 1}  \\ 
 \end{array} } \right) \\&  \ \ \ \ \ \ \times \left[ {\frac{m}
{{\mu \kappa  + m}}} \right]^{ m} \left[ {\frac{{\mu \kappa }}
{{\mu \kappa  + m}}} \right]^{ - m - j + 1} ,  \cr 
  & A_{2j}  = \left( { - 1} \right)^{j - 1} \left( {\begin{array}{c}
   {\mu  - m + j - 2}  \\ 
   {j - 1}  \\ 
 \end{array} } \right) \\&  \ \ \ \ \ \ \times \left[ {\frac{m}
{{\mu \kappa  + m}}} \right]^{j - 1} \left[ {\frac{{\mu \kappa }}
{{\mu \kappa  + m}}} \right]^{m - \mu  - j + 1} ,  \cr 
  & B_j  = \left( {\begin{array}{c}
   {m - \mu }  \\ 
   j  \\ 
 \end{array} } \right)\left[ {\frac{m}
{{\mu \kappa  + m}}} \right]^j \left[ {\frac{{\mu \kappa }}
{{\mu \kappa  + m}}} \right]^{m - \mu  - j} , \cr
\end{split}
\end{equation}
and where we have defined
\begin{equation}
\label{eq:006}
\begin{split}
  & \omega _{A1}  \triangleq \Delta_1 \left( {\mu  - m - j + 1} \right),\quad   \cr 
  & \omega _{A2}  \triangleq \Delta_2 \left( {m - j + 1} \right),  \cr 
  & \  \omega _B  \triangleq \Delta_2 \left( {m - j} \right), \cr
\end{split}
\end{equation}
with
\begin{equation}
\label{eq:007}
\begin{split}
  & \Delta_1  \triangleq \frac{{\bar \gamma }}
{{\mu \left( {1 + \kappa } \right)}},\quad   \cr 
  & \Delta_2  \triangleq \frac{{\mu \kappa  + m}}
{m}\frac{{\bar \gamma }}
{{\mu \left( {1 + \kappa } \right)}}. \cr
\end{split}
\end{equation}

\end{theorem}

\begin{proof}
See Appendix \ref{App1}.
\end{proof}
\vspace{2mm}

Expressions (\ref{eq:003}) and (\ref{eq:004}) give an exact representation of the $\kappa$-$\mu$ shadowed distribution in terms of a finite mixture of squared Nakagami distributions. With this representation of the PDF, which is new in the literature to the best of our knowledge, a similar result for the $\kappa$-$\mu$ shadowed CDF is now obtained.
\vspace{2mm}

\begin{table*}[t]
  \renewcommand{\arraystretch}{3}
\centering
{
\caption{Parameter values for the $\kappa$-$\mu$ shadowed distribution with integer $\mu$ and $m$,}
\label{table02}
{
\begin{tabular}{|c|c|}
\hline\hline
Case $\mu>m$ & Case $\mu \leq m$ \\ \hline\hline 
$M=\mu$ & $M=m-\mu$  \\ \hline
 $C_i=\begin{cases} 
      0 & i=0 \\
       \left( { - 1} \right)^m \binom{m+i-2}{i-1}\times \left[ {\frac{m}
{{\mu \kappa  + m}}} \right]^{ m} \left[ {\frac{{\mu \kappa }}
{{\mu \kappa  + m}}} \right]^{ - m - i + 1}  & 0<i\leq \mu-m \\
      \left( { - 1} \right)^{i-\mu+m - 1} \binom{i-2}{i-\mu+m-1} \times \left[ {\frac{m}
{{\mu \kappa  + m}}} \right]^{i-\mu+m - 1} \left[ {\frac{{\mu \kappa }}
{{\mu \kappa  + m}}} \right]^{-i + 1}  & \mu-m < i \leq \mu 
   \end{cases}$
 & $C_i=\binom{m-\mu}{i}\left[ {\frac{m}
{{\mu \kappa  + m}}} \right]^i \left[ {\frac{{\mu \kappa }}
{{\mu \kappa  + m}}} \right]^{m - \mu  - i}$  \\ \hline
\ $m_i=\begin{cases} \mu-m-i+1, & 0\leq i\leq \mu-m \\
   \mu-i+1 & \mu-m < i \leq \mu 
   \end{cases}$ & $m_i=m-i$  \\ \hline
    $\Omega_i=\begin{cases}  \frac{{\bar \gamma }}
{{\mu \left( {1 + \kappa } \right)}}, & 0\leq i\leq \mu-m \\
   \frac{{\mu \kappa  + m}}
{m}\frac{{\bar \gamma }}
{{\mu \left( {1 + \kappa } \right)}} & \mu-m < i \leq \mu 
   \end{cases}$ & $\Omega_i=\frac{{\mu \kappa  + m}}
{m}\frac{{\bar \gamma }}
{{\mu \left( {1 + \kappa } \right)}}$  \\ \hline
\end{tabular}
}
}
\end{table*}

\begin{corollary} \label{C1}
Let $\gamma$ be a random variable such that $\gamma  \sim \mathcal{S}\left( {\bar \gamma ;\kappa ,\mu ,m} \right)$ and let  $\mu,m \in {\mathbb N}$, then, the CDF of $\gamma$ is given as follows:

$ \bullet$ If $m<\mu$
\begin{equation}
\label{eq:007a}
\begin{split}
  & F_\mathcal{S} \left( {\bar \gamma ;\kappa ,\mu ,m;x} \right) = 1 - \sum\limits_{j = 1}^{\mu  - m} {A_{1j} e^{ - x/\Delta_1 } \sum\limits_{r = 0}^{\mu  - m - j} {\frac{1}
{{r!}}\left( {\frac{x}
{{\Delta_1 }}} \right)^r } }   \cr 
  & \quad \quad \quad  - \sum\limits_{j = 1}^m {A_{2j} e^{ - x/\Delta_2 } \sum\limits_{r = 0}^{m - j} {\frac{1}
{{r!}}\left( {\frac{x}
{{\Delta_2 }}} \right)^r .} }  \cr
\end{split}
\end{equation}

$ \bullet$ If $m \geq\mu$
\begin{equation}
\label{eq:007b}
\begin{split}
F_\mathcal{S} \left( {\bar \gamma ;\kappa ,\mu ,m;x} \right) = 1 - \sum\limits_{t = 0}^{m - \mu } {B_j^{} } e^{ - x/\Delta_2 } \sum\limits_{r = 0}^{m - j - 1} {\frac{1}
{{r!}}\left( {\frac{x}
{{\Delta_2 }}} \right)^r } .
\end{split}
\end{equation}

\end{corollary}

\vspace{2mm}
\begin{proof}
Using Theorem \ref{T1} and considering that the CDF of a squared Nakagami random variable is given by 
\begin{equation}
\label{eq:B1}
\begin{split}
F_\K \left( {\bar \gamma ;\hat m;x} \right) = 1 - e^{ - x/\Delta } \sum\limits_{r = 0}^{ \hat m - 1} {\frac{1}
{{r!}}\left( {\frac{{x}}
{{\Delta }}} \right)^r }, 
\end{split}
\end{equation}
where we have defined $\Delta \triangleq \bar \gamma /\hat m$, the proof is completed.

\end{proof}
{
A more compact and unified form for the PDF and the CDF of the $\kappa$-$\mu$ shadowed fading distribution with integer fading parameters can be obtained after some manipulation, yielding

\begin{align}
\label{eqPDF}
\begin{split}
f_S \left( {\bar \gamma ;\kappa ,\mu ,m;x} \right) = \sum\limits_{i = 0}^{M} {C_i } \frac{x^{m_i-1}}{(m_i-1)!}\frac{1}{\Omega_i^{m_i}}e^{-\frac{x}{\Omega_i}},
\end{split}
\end{align}
and
\begin{align}
\label{eqCDF}
F_\mathcal{S} \left( {\bar \gamma ;\kappa ,\mu ,m;x} \right) = 1-\sum_{i=0}^{M}C_i e^{-\frac{x}{\Omega_i}}\sum_{r=0}^{m_i-1}\frac{1}{r!}\left(\frac{x}{\Omega_i}\right)^{r},
\end{align}
where the parameters $m_i$, $M$ and $\Omega_i$ are expressed in Table \ref{table02} in terms of the parameters of the $\kappa$-$\mu$ shadowed distribution, namely $\kappa$, $\mu$, $m$ and $\bar\gamma$.
}
\section{Discussion}
\label{implications}
After presenting the new results for the statistics of the $\kappa$-$\mu$ shadowed distribution with integer fading parameters, we now discuss about the main implications and insights that can be obtained from these results.
\subsection{Finite mixture representation}
Taking a deeper look at Theorem \ref{T1}, one can observe that the expression for the mixture of squared Nakagami distributions has different form depending on whether $m$ is larger or smaller than $\mu$. 

When $m\geq\mu$ the $\kappa$-$\mu$ shadowed distribution is expressed as a \textit{proper mixture}\footnote{We use the term \textit{proper mixture} to denote any mixture distribution which can be expressed as a convex combination (i.e. a weighted sum with non-negative coefficients that sum to 1) of other distributions. We also use the term \textit{improper mixture} to denote any mixture distribution on which the mixture coefficients are not restricted to be non-negative.} of squared Nakagami distributions, on which the values of the weights $B_j$ in (\ref{eq:005}) correspond to those of the binomial distribution, whose probability mass function is given by
\begin{equation}
f(k;n,p)=\Pr\{X=k\}=\binom{n}{k}p^k(1-p)^{n-k},
\end{equation}
with $n=m-\mu$ and $p=\frac{m}{\mu\kappa+m}$. In this situation, the $\kappa$-$\mu$ shadowed fading model can be regarded as a superposition of $m-\mu$ parallel channels affected by Nakagami fading with different fading severities, being each of such channels used with a given probability{\footnote{We must here note that this phenomenon is purely mathematical, as it arises from the observation of the new form of the PDF here derived. To the best of our knowledge, such observation does not have any connection with the physical model of the $\kappa$-$\mu$ shadowed fading distribution (note that the physical models that originate this distribution \cite{Laureano2016} can be regarded as coherent combinations, in a maximal ratio combining form, of $\mu$ Rician shadowed variates).}} $p$.

Conversely, for $m<\mu$ the PDF in (\ref{eq:003}) is given in terms of an \textit{improper mixture}; i.e., (\ref{eq:003}) is expressed as a linear combination of two different sets of squared Nakagami distributions with coefficients $A_{1j}$ and $A_{2j}$, on which the mixture coefficients are not necessarily non-negative.

Another important insight arises by inspecting eq. (\ref{eq:A3}) in the Appendix. We note that such MGF is expressed as the product of the MGFs of two Gamma distributions with scale parameters $\Delta_1$ and $\Delta_2$, and shape parameters $\mu-m$ and $m$, respectively. Thus, a $\kappa$-$\mu$ shadowed RV can be generated as the sum of two independent gamma RVs with such scale and shape parameters, provided that $m<\mu$. Note that this holds for any $\{m,\mu\}\in\mathbb{R}^+$. For the specific case on which the $\kappa$-$\mu$ shadowed distribution reduces to the $\eta$-$\hat\mu$ distribution, i.e. $\mu=2\hat\mu$ and $m=\hat\mu$ with $\kappa=(1-\eta)/2\eta$ \cite{Laureano2016}, this observation coincides with the one given in \cite{Ermolova2011}.

\subsection{Convergence to the $\kappa$-$\mu$ distribution}

Intuitively, if we let $m\rightarrow\infty$ in the $\kappa$-$\mu$ shadowed fading model, then the Nakagami-$m$ PDF used to model the random fluctuation of the line-of-sight component degenerates to a deterministic distribution, being its PDF the Dirac delta function. Thus, this model reduces to the original $\kappa$-$\mu$ distribution in \cite{Yacoub07}. This implies that by virtue of Theorem \ref{T1}, the $\kappa$-$\mu$ shadowed distribution with integer fading parameters and sufficiently large $m$ can be used to approximate the $\kappa$-$\mu$ distribution. This is observed in Fig. \ref{fo}, where the evolution of the $\kappa$-$\mu$ shadowed distribution as $m$ grows is represented. 

\begin{figure}[t]
\centering
    \includegraphics[width=1.0\columnwidth]{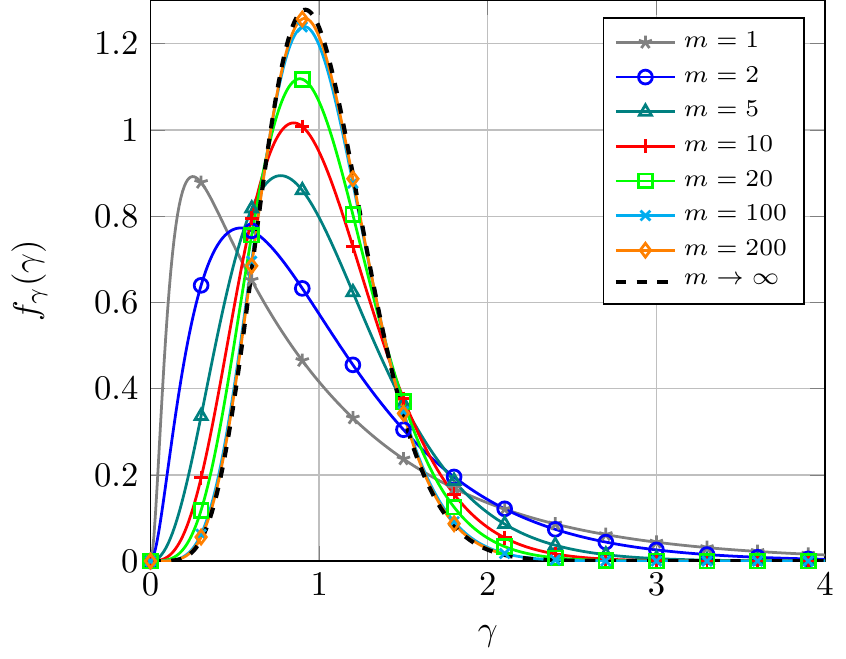}
          \caption{Convergence in distribution between the $\kappa$-$\mu$ shadowed distribution and the $\kappa$-$\mu$ distribution as $m\rightarrow\infty$. Parameter values $\kappa=5$, $\mu=3$ and $\bar\gamma=1$.}
    \label{fo}
  \end{figure}

However, a rigorous proof for the convergence in distribution between the $\kappa$-$\mu$ shadowed and the $\kappa$-$\mu$ fading models is not that evident. 
The following Lemmas formally establish the connections between the $\kappa$-$\mu$ shadowed distribution and the $\kappa$-$\mu$ distribution in terms of weak convergence of probability measures.

\begin{lemma}
\label{lema1p} Let $\{\gamma_m\}_{m=1}^{\infty}$ be a sequence of
random variables with
$\gamma_m\sim\mathcal{S}(\bar\gamma;\kappa,\mu,m)$ and the
corresponding sequence of CDFs given by
$\{F_{\mathcal{S}}(\bar\gamma;\kappa,\mu,m;\gamma)\}_{m=1}^{\infty}$. Then,
this sequence of random variables weakly converges to the
$\kappa$-$\mu$ distribution with mean $\bar\gamma$ and shaping
parameters $\kappa$ and $\mu$, i.e.
\begin{equation}
\label{eq_lema1}
\lim_{m\rightarrow \infty} F_{\mathcal{S}}(\bar\gamma;\kappa,\mu,m;\gamma) = F_{\mathcal{KM}}(\bar\gamma;\kappa,\mu;\gamma),
\end{equation}
where $F_{\mathcal{KM}}$ represents the CDF of the corresponding $\kappa$-$\mu$ distribution.
\end{lemma}

\begin{proof}
It is a direct consequence of the L\'evy's continuity theorem \cite{Gray1996}.
\end{proof}

Lemma \ref{lema1p} establishes that the $\kappa$-$\mu$ shadowed distribution converges to the $\kappa$-$\mu$ distribution for sufficiently large $m$. Thus, it is possible to approximate the $\kappa$-$\mu$ distribution with integer $\mu$ by a finite mixture of squared Nakagami distributions, by using the $\kappa$-$\mu$ shadowed fading distribution with integer fading parameters and choosing an arbitrarily large parameter $m\in\mathbb{N}$.

These connections between the $\kappa$-$\mu$ shadowed distribution and the $\kappa$-$\mu$ distribution can be extended to any probability measure obtained by expectation over the desired distribution, as stated in the following Lemma.
\begin{lemma}
\label{lema2p} Let $\{\gamma_m\}_{m=1}^{\infty}$ be a sequence of
random variables with
$\gamma_m\sim\mathcal{S}(\bar\gamma;\kappa,\mu,m)$ and the
corresponding sequence of CDFs given by
$\{F_{\mathcal{S}}(\bar\gamma;\kappa,\mu,m;\gamma)\}_{m=1}^{\infty}$.
Let $\Phi(\gamma)$ be any probability measure conditioned to $\gamma$, which is
continuous and bounded as a function of $\gamma$.
Then, the sequence of the expectations of $\Phi$ over the $\kappa$-$\mu$ shadowed random variables converges to the
expectation of $\Phi$ over a $\kappa$-$\mu$ random variable with mean $\bar\gamma$ and shaping
parameters $\kappa$ and $\mu$, i.e.
\begin{equation}
\label{eq_lema1}
\lim_{m\rightarrow \infty} \int_0^{\infty}{\kern -2mm} \Phi(\gamma) d F_{\mathcal{S}}(\bar\gamma;\kappa,\mu,m;\gamma)
= \int_0^{\infty}{\kern -2mm} \Phi(\gamma) d F_{\mathcal{KM}}(\bar\gamma;\kappa,\mu;\gamma),
\end{equation}
where $F_{\mathcal{KM}}$ represents the CDF of the corresponding $\kappa$-$\mu$ distribution.
\end{lemma}

\begin{proof}
It is a direct consequence of Lemma \ref{lema1p} and the Helly-Bray theorem \cite{Billy}.
\end{proof}
%

\subsection{A new gamma approximation to the Rician distribution}

In the milestone paper by Nakagami \cite{Nakagami1960}, a simple equivalence between the Nakagami-$\hat m$ distribution and the Rician distribution was proposed. The connection between these distributions is established by a simple parameter transformation, setting $\hat m = (1+K)^2/(1+2K)$. This approximation, also included in the reference textbook by Simon and Alouini \cite[eq. (2.26)]{AlouiniBook}, has been widely employed in wireless communications in order to approximate the Rician distribution by the more tractable Nakagami-$\hat m$ distribution. 

However, as argued by many authors \cite{Wang2003,Beaulieu2014} such approximation has a severe flaw that has important impact when analyzing the system performance: the diversity order of Rician fading equals one, whereas the diversity order of Nakagami-$\hat m$ fading is $\hat m$. The diversity order is related to the behavior of the PDF around the origin, or equivalently to the asymptotic behavior of the MGF as $s\rightarrow\infty$.

As consequences of Lemma \ref{lema1p}, and setting $\mu=1$, we propose to approximate the Rician distribution by a finite mixture of squared Nakagami distributions, by using the Rician shadowed fading distribution and choosing an arbitrarily large parameter $m\in\mathbb{N}$. The main benefit of this approximation relies on the fact that the diversity order of the Rician shadowed distribution is also $1$.

In Figs. \ref{fx1} and \ref{fx2}, the behavior of the classical approximation to the Rician distribution proposed by Nakagami, and the one here proposed based on the Rician shadowed distribution with integer $m$ is illustrated, representing the corresponding PDFs for different values of $K$ and $m$. Log-log scale was used in order to better observe the effect of increasing $m$. We can see that for low values of $K$, the approximation to the Rician distribution based on the Rician shadowed distribution is good even for moderate values of $m$. As $K$ grows, a larger number of terms is required for the mixture approximation (i.e., a larger $m$) in order to converge to the Rician distribution. In both cases, the smoothness of the PDFs in the proximity of zero for the Rician and Rician shadowed distributions have similar shape, whereas the original approximation in \cite{Nakagami1960} exhibits a very different behavior.

\begin{figure}[t]
\centering
    \includegraphics[width=1.0\columnwidth]{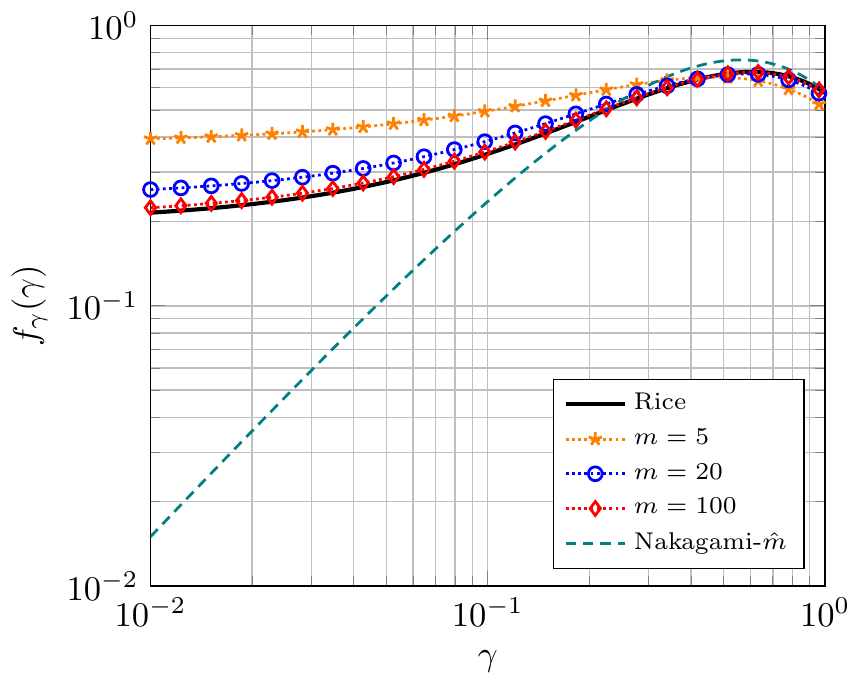}
          \caption{Gamma approximation to the Rician distribution with parameter $K$ using the Rician shadowed distribution with parameter $K$ and integer $m$. Parameter values $K=3$ and $\bar\gamma=1$. Nakagami approximation \cite{Nakagami1960} uses $\hat m = (1+K)^2/(1+2K)=2.28$.}    \label{fx2}
  \end{figure}
  
\begin{figure}[t]
\centering
    \includegraphics[width=1.0\columnwidth]{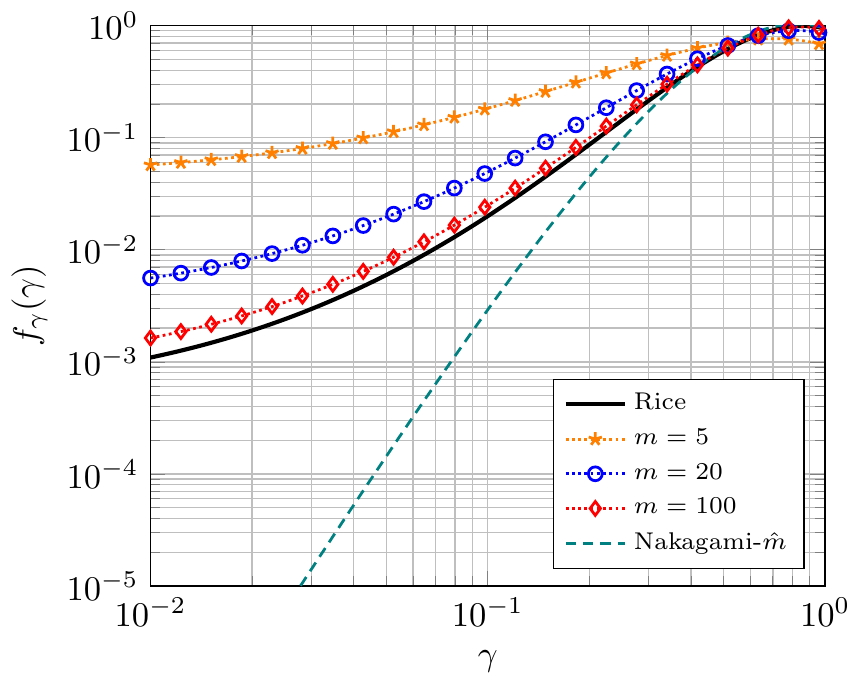}
          \caption{Gamma approximation to the Rician distribution with parameter $K$ using the Rician shadowed distribution with parameter $K$ and integer $m$. Parameter values $K=10$ and $\bar\gamma=1$. Nakagami approximation \cite{Nakagami1960} uses $\hat m = (1+K)^2/(1+2K)=5.76$.}
    \label{fx1}
  \end{figure}

\subsection{Effect of considering integer fading parameters}

The parameter $\mu\in\mathbb{N}$ was originally introduced by Yacoub \cite{Yacoub07} as the number of clusters of multipath waves that propagate in a non-homogeneous environment. As argued in \cite{Yacoub07}, the restriction of the parameter $\mu$ to take integer values is inherently linked to the underlying physical model for the $\kappa$-$\mu$ distribution. Thus, the $\eta$-$\mu$ and $\kappa$-$\mu$ fading models with integer $\mu$ are usually referred to as \textit{physical models}, and are often used to evaluate the performance of communication systems operating over generalized fading channels \cite{Morales2010,Peppas2013,Ermolova2011,Ermolova2014,Tsiftsis2016,Pena2015}. The consideration of a real-valued $\mu$ indeed yields a larger flexibility to the model; however, the impact of considering an integer $\mu$ decreases as $\mu$ grows as observed in Fig. \ref{f1}.

The restriction of the parameter $m$ to take integer values only has a non-negligible impact in heavy shadowing environments (i.e. low values of $m$). However, as $m$ grows the PDFs corresponding to the real-valued $m$ and its closest integer (i.e., largest previous or smallest following integer) counterpart tend to be indistinguishable; this is observed in Fig. \ref{f2}.

\begin{figure}[t]
\centering
    \includegraphics[width=1.0\columnwidth]{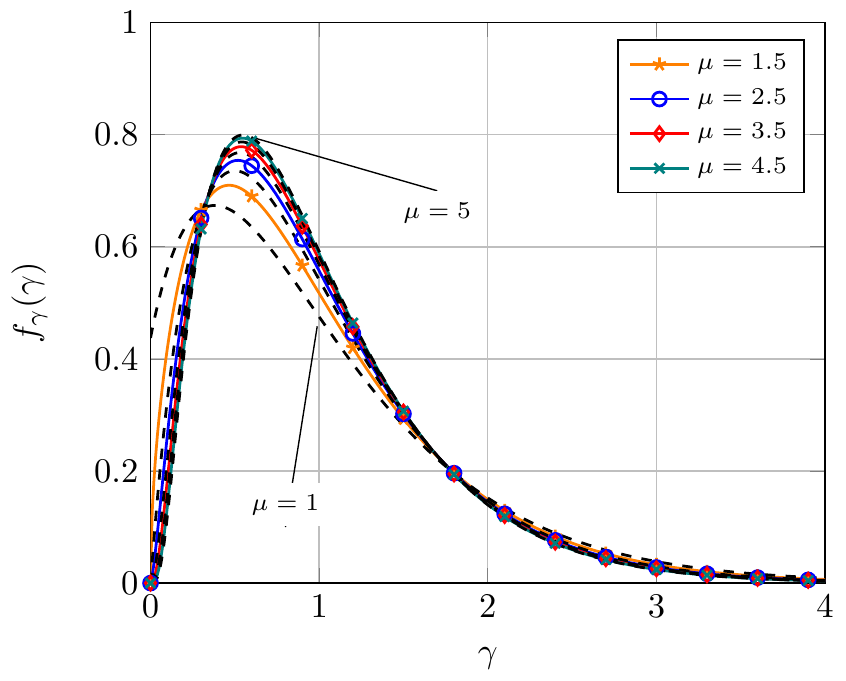}
          \caption{Evolution of the $\kappa$-$\mu$ shadowed fading PDF for different values of $\mu$. Solid lines correspond to real-valued $\mu\in\{1.5,2.5,3.5,4.5\}$, whereas dashed lines correspond to the largest previous and the smallest following integer values $\mu\in\{1:5\}$. Parameter values $\kappa=6$, $m=2$ and $\bar\gamma=1$.}
    \label{f1}
  \end{figure}

\begin{figure}[t]
\centering
    \includegraphics[width=1.0\columnwidth]{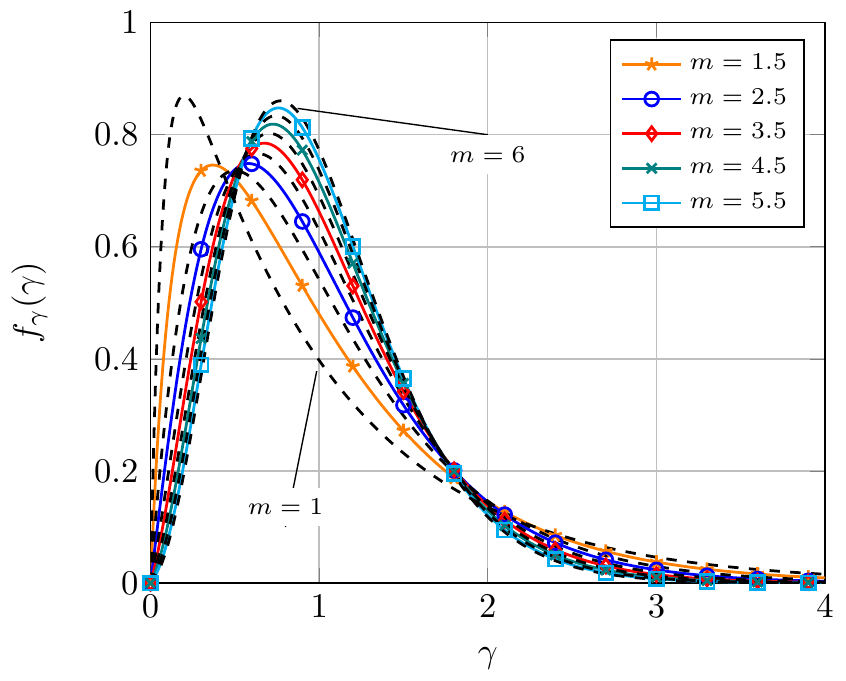}
          \caption{Evolution of the $\kappa$-$\mu$ shadowed fading PDF for different values of $m$. Solid lines correspond to real-valued $m\in\{1.5,2.5,3.5,4.5,5.5\}$, whereas dashed lines correspond to the largest previous and the smallest following integer values $m\in\{1:6\}$. Parameter values $\kappa=6$, $\mu=2$ and $\bar\gamma=1$.}
    \label{f2}
  \end{figure}

{ Thus, from the observation of Figs. \ref{f1} and \ref{f2} we see that the effect of restricting the fading severity parameters to take integer values is limited, unless $m$ or $\mu$ take low values (i.e., severe LOS fluctuation or severe multipath, respectively). In fact, one may wonder whether there's a lower value of $m$ (or equivalently $\mu$) under which the use of the $\kappa$-$\mu$ shadowed fading distribution with integer fading parameters to approximate the $\kappa$-$\mu$ shadowed fading distribution with arbitrary fading parameters is not recommended. In our view, this will strongly depend on the application. Sometimes only the tail of the distributions is needed for further calculations, or even the performance metrics of interest end up being rather similar when some fading severity parameters vary. Thus, it may occur that the approximation error is still negligible for low values of $m$ and $\mu$ in some cases; this will be later exemplified in Section IV.}

We will now study the impact of restricting the fading severity parameters $\mu$ and $m$ to take integer values on the goodness of fit to field measurements. We use the empirical results presented in \cite{Francis2016} for some underwater acoustic channels (UAC), for which the $\kappa$-$\mu$ shadowed fading model showed the best fit. { These measurements were conducted in the Mediterranean Sea near Cartagena (Spain), in shallow waters with depths in the range of 14-30m. Details on the specific measurement configuration, including a block diagram of the measurement equipment set-up, can be found in \cite[Sect. 3.1]{Francis2016}}.

We used a modified version of the Kolmogorov-Smirnov (KS) statistic in order define the error factor $\epsilon$ that is used to quantify the goodness of fit between the empirical and theoretical CDFs, which are denoted by $\hat{F}_r(\cdot)$ and ${F}_r(\cdot)$ respectively, i.e,
\begin{equation}
\epsilon\triangleq \max_{x}|\log_{10} \hat{F}_r(x)-\log_{10} F_r(x)|.
\end{equation}

As in \cite{Francis2016}, the CDF is used in log-scale with the aim of outweighing the fit in those values of received power closer to zero, i.e. those corresponding to a more severe fading. In Figs \ref{fit1} and \ref{fit2} we compare the set of measurements corresponding to the channels C9-32 and C9-64 with three different distributions: Rician, $\kappa$-$\mu$ shadowed with integer fading parameters, and $\kappa$-$\mu$ shadowed. For the C9-32 channel, we observe that the error factor increases from $\epsilon=0.026$ to $\epsilon=0.083$ when constraining $\mu$ and $m$ to take integer values. However, the fit is still improved when comparing to Rician fading, which yields $\epsilon=0.111$.

\begin{figure}[t]
\centering
    \includegraphics[width=1.0\columnwidth]{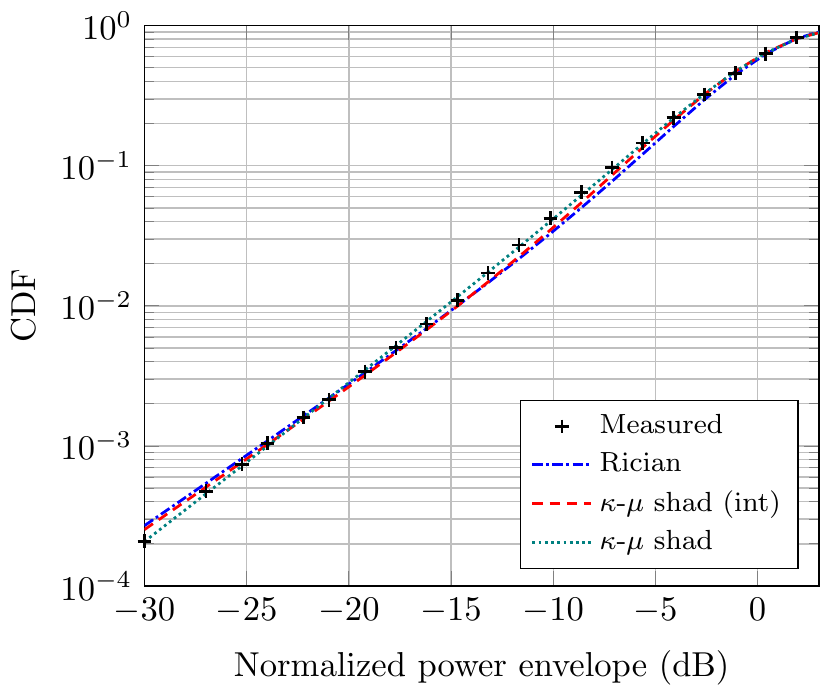}
\caption{Empirical vs theoretical CDFs of the received signal power for the UAC channel C9-32 \cite{Francis2016}. Parameter values: Rician $\{K=2.64;\epsilon=0.111\}$, $\kappa$-$\mu$ shadowed int $\{\kappa=12.84,\mu=1,m=2;\epsilon=0.083\}$; $\kappa$-$\mu$ shadowed $\{\kappa=4.06,\mu=1.13,m=2.45;\epsilon=0.026\}$.}
 \label{fit1}
 \end{figure}

 \begin{figure}[t]
 \centering
    \includegraphics[width=1.0\columnwidth]{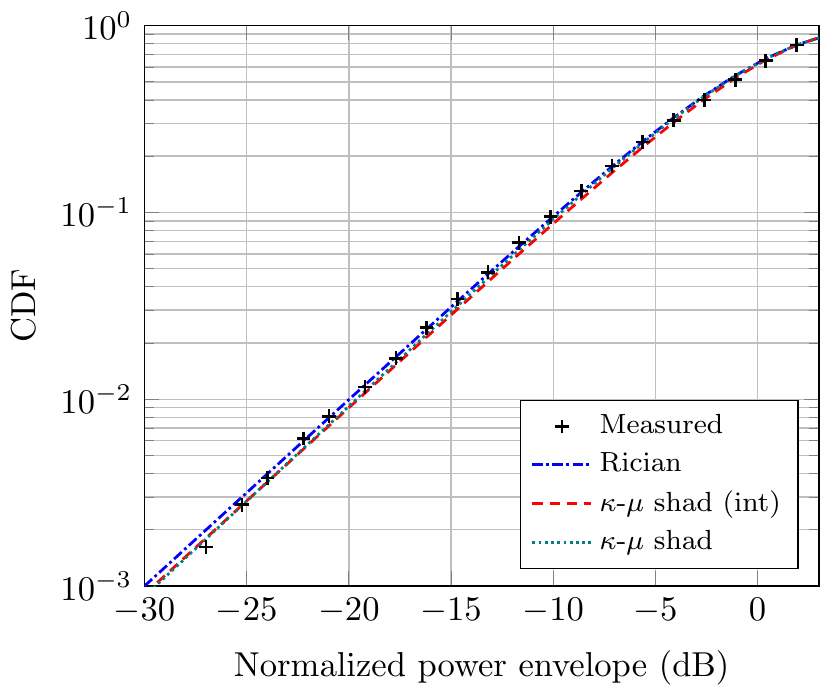}
\caption{Empirical vs theoretical CDFs of the received signal power for the UAC channel C9-64 \cite{Francis2016}. Parameter values: Rician $\{K=0.03;\epsilon=0.104\}$, $\kappa$-$\mu$ shadowed int $\{\kappa=0.6,\mu=1,m=6;\epsilon=0.060\}$; $\kappa$-$\mu$ shadowed $\{\kappa=0.03,\mu=1.02,m=6.32;\epsilon=0.053\}$.}
\label{fit2} 
\end{figure}

With regard to the C9-64 channel, we now observe that the error factor is barely increased (i.e. $\epsilon=0.060$ instead of $\epsilon=0.053$) when forcing $\mu$ and $m$ to take integer values, compared to the general $\kappa$-$\mu$ shadowed fading. For the case of Rician fading, we have $\epsilon=0.104$, which is clearly outperformed by the $\kappa$-$\mu$ shadowed model with integer fading parameters.

%

\subsection{Performance analysis}
The calculation of performance metrics of interest such as error probability or channel capacity requires to integrate over the PDF or CDF of the SNR. Using the  general representation of the $\kappa$-$\mu$ shadowed fading distribution implies integrating over the hypergeometric functions that define its PDF and CDF. However, with the new representation for the $\kappa$-$\mu$ shadowed distribution functions, the performance analysis is greatly simplified. In the following Lemma, we show that analyzing the performance in $\kappa$-$\mu$ shadowed fading has the same complexity as analyzing the much simpler Nakagami-$\hat m$ case.

\begin{lemma} \label{C2}
Let $h(\gamma)$ be a performance metric depending on the instantaneous SNR $\gamma$, and let $\overline{h}_  \K(\overline{\gamma};m)$ be the metric in Nakagami fading with average SNR $\overline{\gamma}$ obtained by averaging over an interval of the PDF of the SNR, i.e.,
\begin{equation} \label{eq:008}
	 \overline{h}_\K(\overline{\gamma};m)= \int_a^b 
	 h(x) f_\K \left( {\bar \gamma ;m;x} \right)	 dx,
\end{equation}
with $0 \leq a < b \leq \infty$.
Then, the average  performance metric in $\kappa$-$\mu$ shadowed fading channels with average SNR $\overline{\gamma}$, denoted as $\overline{h}_\mathcal{S}(\bar \gamma ;\kappa ,\mu ,m)$, can be calculated, given that $\mu,m \in {\mathbb N}$, as
{
\begin{equation} \label{eq:009}
\begin{split}
  \overline h _S \left( {\bar \gamma ;\kappa ,\mu ,m} \right) = \sum\limits_{i = 0}^{M} {C_i } \overline h _K \left( {\Omega_i/m_i;m_i} \right),
\end{split}
\end{equation}
}
\end{lemma}
\begin{proof}
{
The average performance metric $\overline{h}_\mathcal{S}(\bar \gamma ;\kappa ,\mu ,m)$ is calculated as
\begin{equation} \label{eq:0C1}
	 \overline{h}_\mathcal{S}(\bar \gamma ;\kappa ,\mu ,m)= \int_a^b 
	 h(x) f_\mathcal{S}\left( {\bar \gamma ;\kappa ,\mu ,m;x} \right)  dx.
\end{equation}
We can express $f_\mathcal{S}\left( {\bar \gamma ;\kappa ,\mu ,m;x} \right)$ as a mixture of squared Nakagami distributions, given in (\ref{eqPDF}). Thus, combining (\ref{eq:002}) with (\ref{eqPDF}) we have
\begin{equation}
\label{eq:0C2}
\begin{split}
  & f_\mathcal{S} \left( {\bar \gamma ;\kappa ,\mu ,m;x} \right) = \sum\limits_{i = 0}^{M} {C_i} f_K \left( {\Omega_i/m_i ;m_i;x} \right).
\end{split}
\end{equation}
Introducing (\ref{eq:0C2}) into (\ref{eq:0C1}), after some simple manipulation, (\ref{eq:009}) is obtained.
}
\end{proof}

The implications of this lemma are of great importance, as it means that \textit{any performance metric} for which existing results are available for the Nakagami-$\hat m$ case can be directly generalized to the $\kappa$-$\mu$ shadowed case by means of a finite linear combination as described in (\ref{eq:008}). We will exemplify the usefulness of this Lemma in the next Section.

\section{Application: Average channel capacity}
\label{apps}
%
The characterization of the average channel capacity in fading channels, defined as 
\begin{equation}
\bar{C}[\text{bps/Hz}]\triangleq\int_0^{+\infty}\log_2(1+\gamma)f_\gamma(\gamma)d\gamma,
\end{equation}
where $\gamma$ is the instantaneous SNR at the receiver side, is a classical problem in communication theory \cite{Lee1990,Gunther1996,Alouini1999}. The average channel capacity assuming $\kappa$-$\mu$ shadowed fading was derived in \cite{Celia2014}, in terms of the unwieldy bivariate Meijer-$G$ function. 

However, a direct application of Lemma \ref{C2} using the average channel capacity expression for Nakagami-$\hat m$ fading channels \cite[eq. (23)]{alouini1997} yields the following simple closed-form expression:
{
\begin{equation} \label{eq:013}
\begin{split}
  & \overline C  = \log _2 \left( e \right) \sum\limits_{i = 0}^{M} {C_i }e^{1/\Omega_i } \sum\limits_{k = 0}^{m_i- 1} {\frac{{\Gamma \left( { - k,1/\Omega_i } \right)}}
{{\Omega_i^k }}} , \cr
\end{split}
\end{equation}
}

where $\Gamma(\cdot)$ is the upper incomplete Gamma function, which can be computed, when the first parameter is a negative integer, as \cite[eq. (8.352.3)]{Gradstein2007}
\begin{equation} \label{eq:011}
\begin{split}
\Gamma ( - n,x) = \frac{{\left( { - 1} \right)^n }}
{{n!}}\left[ { - Ei\left( { - x} \right) - e^{ - x} \sum\limits_{r = 0}^{n - 1} {\left( { - 1} \right)^r \frac{{r!}}
{{x^{r + 1} }}^r } } \right],
\end{split}
\end{equation}
and 
\begin{equation}
\label{eq:012}
\Gamma (0,x) =  - Ei\left( { - x} \right),
\end{equation}
where $Ei(\cdot)$ is the exponential integral function \cite[eq. (8.211.1)]{Gradstein2007}.

In the next set of figures, the effect of the fading severity parameters $\mu$ and $m$ on the average capacity is investigated. Firstly, Figs. \ref{f5} and \ref{f6} illustrate the effect of $m$ in different conditions: strong LOS ($\kappa=10$) and weak LOS ($\kappa=1$), assuming a fixed value of $\mu=3$. In general terms, a larger $m$ is translated into a larger capacity for a given SNR. However, the effect of $m$ is much more pronounced in the strong LOS scenario, leading to a more severe performance degradation for lower $m$ (i.e. heavy shadowing in the LOS component.). Conversely, in the weak LOS scenario the effect of $m$ is barely noticeable.

%

\begin{figure}[t]
\centering
    \includegraphics[width=1.0\columnwidth]{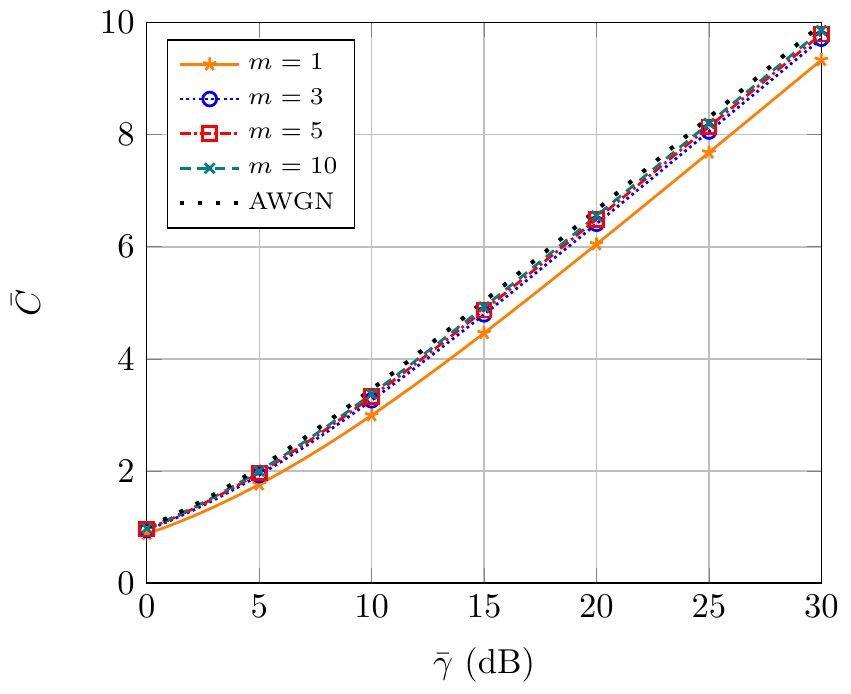}
          \caption{Average channel capacity vs. average SNR, for different values of $m$. Parameter values are $\kappa=10$ and $\mu=3$. The AWGN case is included as a reference.}
    \label{f5}
  \end{figure}

\begin{figure}[t]
\centering
    \includegraphics[width=1.0\columnwidth]{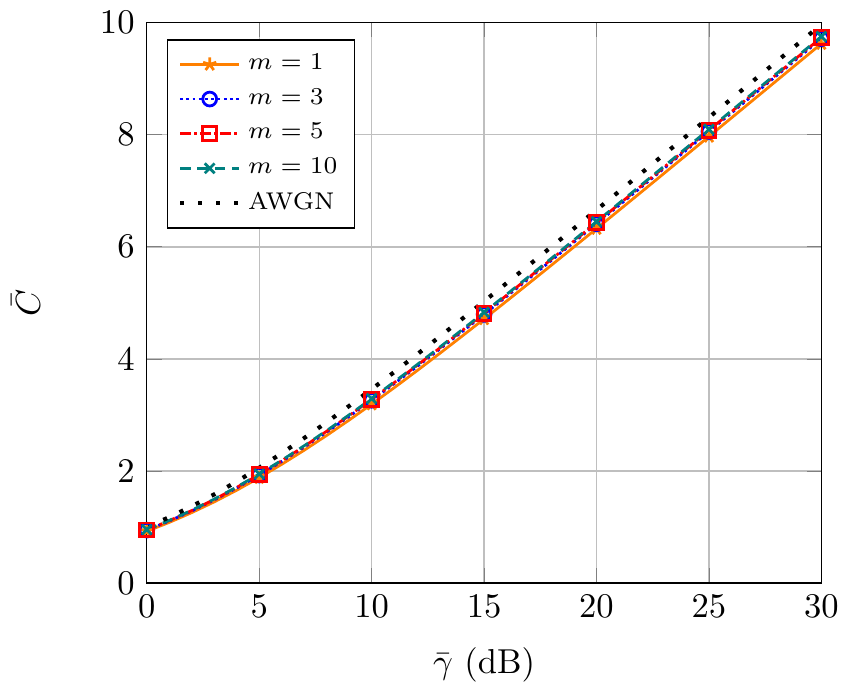}
          \caption{Average channel capacity vs. average SNR, for different values of $m$. Parameter values are $\kappa=1$ and $\mu=3$. The AWGN case is included as a reference.}
    \label{f6}
  \end{figure}

In Figs. \ref{f7} and \ref{f8}, the effect of $\mu$ is investigated in the same conditions as in the previous case: strong LOS ($\kappa=10$) and weak LOS ($\kappa=1$), assuming a fixed value of $m=3$. We see that in the strong LOS scenario, increasing $\mu$ has little effect as the performance is dominated by the LOS component. Conversely, increasing the number of clusters $\mu$ in the weak LOS scenario is translated into a better performance.

\begin{figure}[t]
\centering
    \includegraphics[width=1.0\columnwidth]{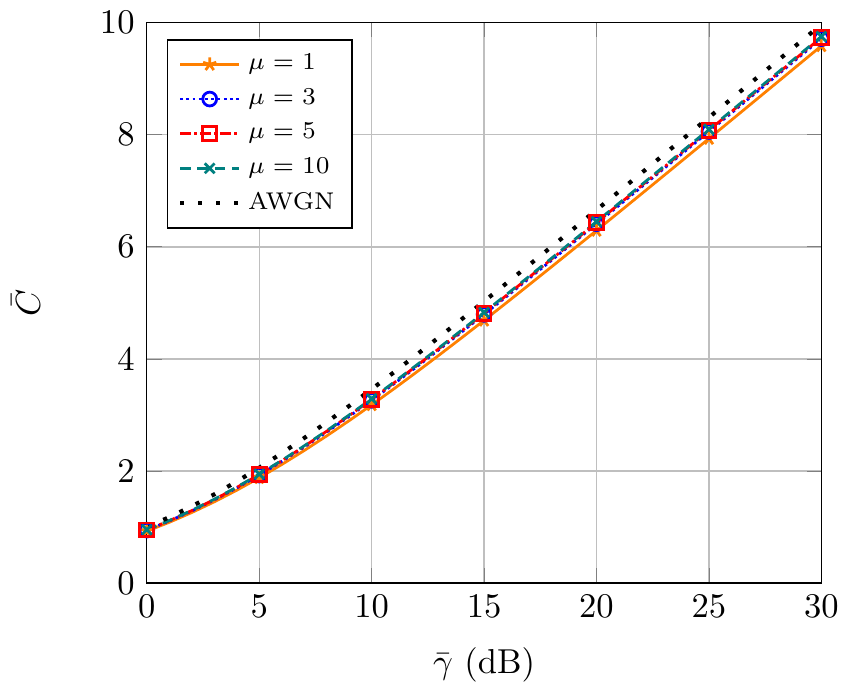}
          \caption{Average channel capacity vs. average SNR, for different values of $\mu$. Parameter values are $\kappa=10$ and $m=3$. The AWGN case is included as a reference.}
    \label{f7}
  \end{figure}

\begin{figure}[t]
\centering
    \includegraphics[width=1.0\columnwidth]{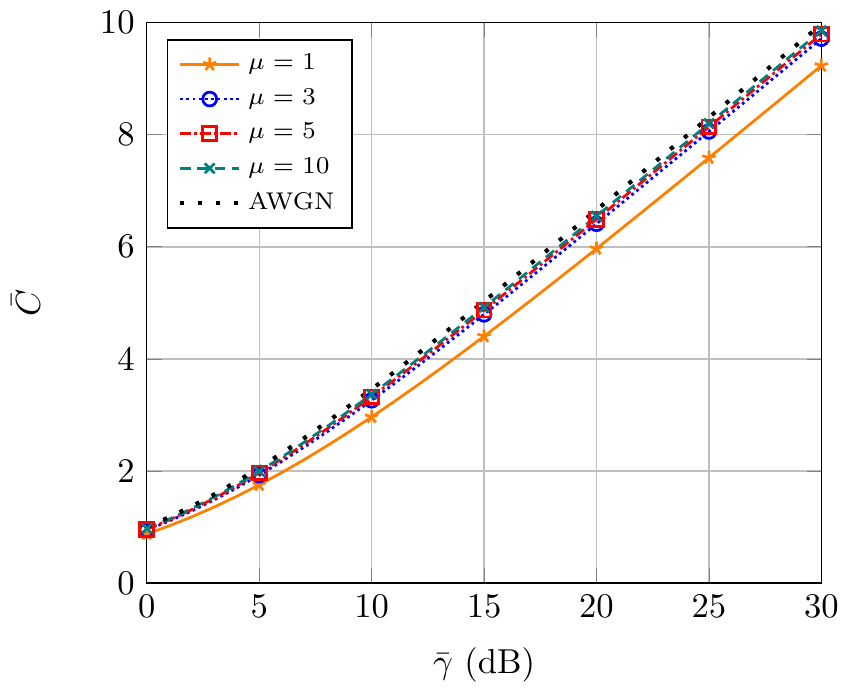}
          \caption{Average channel capacity vs. average SNR, for different values of $\mu$. Parameter values are $\kappa=1$ and $m=3$. The AWGN case is included as a reference.}
    \label{f8}
  \end{figure}

\section{Conclusions}
\label{conc}

The statistical characterization of the $\kappa$-$\mu$ shadowed model with integer fading parameters was here presented. Remarkably, the PDF and CDF of this very general model can be expressed in closed-form in terms of a finite sum of powers and exponentials, in the form of a mixture of Gamma distributions. The inherent mathematical complexity of the $\kappa$-$\mu$ shadowed fading model is greatly reduced at the expense of a limited restriction in terms of flexibility. Most notably, the performance evaluation of communication systems operating under this fading channel model can be directly evaluated if existing results are available for the simpler Nakagami-$\hat m$ fading model. Thus, considering the $\kappa$-$\mu$ shadowed fading model implies no additional complexity while allowing for a much larger flexibility in terms of propagation conditions, naturally including both LOS and NLOS scenarios.

\section{Acknowledgement}

This work has been funded by the Consejer\'ia de Econom\'ia, Innovaci\'on, Ciencia y Empleo of the Junta de Andaluc\'ia, the Spanish Government
and the European Fund for Regional Development FEDER (projects P2011-TIC-7109, P2011-TIC-8238, TEC2013-42711-R, TEC2013-44442-P and TEC2014-57901-R). The first author would like to thank Marco di Renzo for enlightening discussion about the diversity order of the approximations to the Rician distribution.

\appendices
\section{Proof of Theorem I}
\label{App1}
The MGF of  a $\kappa$-$\mu$ shadowed random variable is given by \cite{Paris2014}
\begin{equation}
\label{eq:A1}
\begin{split}
\M\left( s \right) = \frac{{\left( { - \mu } \right)^\mu  m^m \left( {1 + \kappa } \right)^\mu  }}
{{\bar \gamma ^\mu  \left( {\mu \kappa  + m} \right)^m }}\frac{{\left( {s - \frac{{\mu \left( {1 + \kappa } \right)}}
{{\bar \gamma }}} \right)^{m - \mu } }}
{{\left( {s - \frac{{\mu \left( {1 + \kappa } \right)}}
{{\bar \gamma }}\frac{m}
{{\mu \kappa  + m}}} \right)^m }},
\end{split}
\end{equation}
which can be written, in terms of $\Delta_1$ and $\Delta_2$ defined in (\ref{eq:007}), as
\begin{equation}
\label{eq:A2}
\M\left( s \right) = \frac{{\left( {1 - \Delta_1 s} \right)^{m - \mu } }}
{{\left( {1 - \Delta_2 s} \right)^m }}.
\end{equation}
The PDF is related to the MGF by the inverse Laplace transform $ \L^{-1}\{\M(-s);s;x\}$, therefore, it is clear that its analytical expression will depend on the exponent of the numerator being positive (or zero), i.e. $m \geq \mu$, or negative, i.e. $m < \mu$.

Let us consider $m < \mu$. The MGF can be rewritten in a more convenient way for this case as
\begin{equation}
\label{eq:A3}
\M\left( s \right) = \frac{1}
{{\left( {1 - \Delta_1 s} \right)^{\mu  - m} }}\frac{1}
{{\left( {1 - \Delta_2 s} \right)^m }},
\end{equation}
and performing a partial fraction expansion we obtain
\begin{equation}
\label{eq:A4}
\M\left( s \right) = \sum\limits_{j = 1}^{\mu  - m} {\frac{{A_{1j} }}
{{\left( {1 - \Delta_1 s} \right)^{\mu  - m - j + 1} }}}  + \sum\limits_{j = 1}^m {\frac{{A_{2j} }}
{{\left( {1 - \Delta_2 s} \right)^{m - j + 1} }}} ,
\end{equation}
where
\begin{equation}
\label{eq:A5}
\begin{split}
  & A_{1j}  =  \left( { - 1} \right)^{j - 1} \left( {\begin{array}{c}
   {m + j - 2}  \\ 
   {j - 1}  \\ 
 \end{array} } \right) 
\frac{{\Delta_2^{j - 1} \Delta_1^m }}
{{\left( {\Delta_1  - \Delta_2 } \right)^{m + j - 1} }},  \cr 
  & A_{2j}  = \left( { - 1} \right)^{j - 1} \left( {\begin{array}{c}
   {\mu  - m + j - 2}  \\ 
   {j - 1}  \\ 
 \end{array} } \right) 
\frac{{\Delta_1^{j - 1} \Delta_2^{\mu  - m} }}
{{\left( {\Delta_2  - \Delta_1 } \right)^{\mu  - m + j - 1} }}.  \cr 
\end{split}
\end{equation}
By plugging the defined $\Delta_1$ and $\Delta_2$ given in (\ref{eq:007}) into (\ref{eq:A5}), the expressions of coefficients $A_{1j}$ and $A_{2j}$ given in  (\ref{eq:005}) are obtained after some algebraic manipulation.
Performing now an inverse Laplace transformation to the MGF as given in (\ref{eq:A4}),  (\ref{eq:003}) is obtained.

Let us now consider $m \geq \mu$ and let us rewrite the MGF expression given in (\ref{eq:A2}) as 
\begin{equation}
\label{eq:A6}
\begin{split}
\M\left( s \right) = \frac{{\Delta_1^{m - \mu } }}
{{\Delta_2^m }}\frac{{\left( {\frac{1}
{{\Delta_1 }} - s} \right)^{m - \mu } }}
{{\left( {\frac{1}
{{\Delta_2 }} - s} \right)^m }}.
\end{split}
\end{equation}
The PDF can be obtained from the MGF given in (\ref{eq:A6}) and by employing the derivative and the modulation properties of the Laplace transformation as follows:
\begin{equation}
\label{eq:A7}
\begin{split}
  & f_\mathcal{S} \left( {\bar \gamma ;\kappa ,\mu ,m;x} \right) = \L^{ - 1} \left\{ {\M\left( { - s} \right);s.x} \right\}  \cr 
  & \quad  = \frac{{\Delta_1^{m - \mu } }}
{{\Delta_2^m }}\L^{ - 1} \left\{ {\frac{{\left( {\frac{1}
{{\Delta_1 }} + s} \right)^{m - \mu } }}
{{\left( {\frac{1}
{{\Delta_2 }} + s} \right)^m }};s.x} \right\}  \cr 
  & \quad  = e^{ - x/\Delta_1 } \frac{{\Delta_1^{m - \mu } }}
{{\Delta_2^m }}\L^{ - 1} \left\{ {\frac{{s^{m - \mu } }}
{{\left( {\frac{1}
{{\Delta_2 }} - \frac{1}
{{\Delta_1 }} + s} \right)^m }};s.x} \right\}  \cr 
  & \quad  = e^{ - x/\Delta_1 } \frac{{\Delta_1^{m - \mu } }}
{{\Delta_2^m }}\Delta_3^m \frac{{d^{m - \mu } }}
{{dx^{m - \mu } }}\L^{ - 1} \left\{ {\frac{1}
{{\left( {1 - \Delta_3 s} \right)^m }};s.x} \right\} \cr
  & \quad  = e^{ - x/\Delta_1 } \frac{{\Delta_1^{m - \mu } }}
{{\Delta_2^m \left( {m - 1} \right)!}}\frac{{d^{m - \mu } }}
{{dx^{m - \mu } }}x^{m - 1} e^{x/\Delta_3 } , \cr
\end{split}
\end{equation}
where we have defined
\begin{equation}
\label{eq:A8}
\begin{split}
\Delta_3  \triangleq \frac{{\Delta_1 \Delta_2  }}
{{\Delta_2  - \Delta_1 }}.
\end{split}
\end{equation}
From the Leibniz derivative rule we can write
\begin{equation}
\label{eq:A9}
\begin{split}
  & f_\mathcal{S} \left( {\bar \gamma ;\kappa ,\mu ,m;x} \right) = e^{ - x/\Delta_1 } \frac{{\Delta_1^{m - \mu } }}
{{\Delta_2^m \left( {m - 1} \right)!}}\sum\limits_{r = 0}^{m - \mu } {\left( {\begin{array}{c}
   {m - \mu }  \\ 
   r  \\ 
 \end{array} } \right)}   \cr 
   & \quad \quad \quad  \times \left( {\frac{{d^r }}
{{dx^r }}x^{m - 1} } \right)\left( {\frac{{d^{m - \mu  - r} }}
{{dx^{m - \mu  - r} }}e^{x/\Delta_3 } } \right)  \cr 
  & \quad \quad \quad  = e^{ - x/\Delta_1 } \frac{{\Delta_1^{m - \mu } }}
{{\Delta_2^m \left( {m - 1} \right)!}}\sum\limits_{t = 0}^{m - \mu } {\left( {\begin{array}{c}
   {m - \mu }  \\ 
   r  \\ 
 \end{array} } \right)}   \cr 
  & \quad \quad \quad  \times \frac{{\left( {m - 1} \right)!}}
{{\left( {m - 1 - r} \right)!}}x^{m - 1 - r} \left( {\frac{1}
{{\Delta_3 ^{m - \mu  - r} }}e^{x/\Delta_3 } } \right). \cr
\end{split}
\end{equation}
where we have considered that $r \leq m-1$ to calculate the $r^{th}$ order derivative of $x^{m-1}$, thus
\begin{equation}
\label{eq:A10}
\begin{split}
  & f_\mathcal{S} \left( {\bar \gamma ;\kappa ,\mu ,m;x} \right)  = \sum\limits_{t = 0}^{m - \mu } {\left( {\begin{array}{c}
   {m - \mu }  \\ 
   r  \\ 
 \end{array} } \right)\frac{{\Delta_1^{m - \mu } }}
{{\Delta_2^r \Delta_3^{m - \mu  - r} }}}   \cr 
  & \quad \quad \quad  \times \frac{1}
{{\left( {m - 1 - r} \right)!\Delta_2^{m - r} }}x^{m - 1 - r} e^{ - x/\Delta_2 } . \cr
\end{split}
\end{equation}
From the definitions of $\Delta_1$, $\Delta_2$ and $\Delta_3$ given in (\ref{eq:007}) and (\ref{eq:A8}), and noting that 
\begin{equation}
\label{eq:A11}
\begin{split}
  B_j= {\left( {\begin{array}{c}
   {m - \mu }  \\ 
   r  \\ 
 \end{array} } \right)\frac{{\Delta_1^{m - \mu } }}
{{\Delta_2^r \Delta_3^{m - \mu  - r} }}},
\end{split}
\end{equation}

after some algebraic manipulation, (\ref{eq:004}) is finally obtained.

\bibliographystyle{IEEEtran}
\bibliography{bibfile}

\end{document}